\documentclass[12pt]{amsart}
\usepackage{amsmath}
\usepackage{amssymb}
\usepackage{latexsym}

\newcommand{\ignore}[1]{}
\numberwithin{equation}{section}
\newcommand{\Romannumeral}[1]{\uppercase\expandafter{\romannumeral#1}}

\newcommand{\ai}{\text{I}}
\newcommand{\ait}{\text{II}}
\newcommand{\aitr}{\text{III}}
\newcommand{\nul}{\text{N}}

\newcommand{\Z }{\mathbb {Z}}    
     
\newcommand{\Rea}{\mathbb {R}}      

\newtheorem{prop}{\bf Proposition}[section]
\newtheorem{thm}[prop]{\bf Theorem}

\begin{document}

\title[Multi-indexed Jacobi polynomials and Maya diagrams]
{Multi-indexed Jacobi polynomials and Maya diagrams}
\author{Kouichi Takemura}
\address{Department of Mathematics, Faculty of Science and Technology, Chuo University, 1-13-27 Kasuga, Bunkyo-ku Tokyo 112-8551, Japan.}
\email{takemura@math.chuo-u.ac.jp}
\subjclass[2010]{33C45,81Q80}
\begin{abstract}
Multi-indexed Jacobi polynomials are defined by the Wronskian of four types of eigenfunctions of a deformed  P\"oschl-Teller Hamiltonian.
We give a correspondence between multi-indexed Jacobi polynomials and pairs of Maya diagrams, and we show that any multi-indexed Jacobi polynomial is essentially equal to some multi-indexed Jacobi polynomial of two types of eigenfunction.
As an application, we show a Wronskian-type formula of some special eigenstates of the deformed  P\"oschl-Teller Hamiltonian.
\end{abstract}

\maketitle

\section{Introduction}
\label{sec:dpt}

Recently systems of orthogonal polynomials which are out of range of Bochner's theorem are studied actively (see \cite{GKM,MQ,HST,Odake} and references therein).
A typical example of them is the multi-indexed Jacobi polynomials, which is a  generalizations of exceptional Jacobi polynomials.
Here we recall the systems of quantum mechanics which are related to the multi-indexed Jacobi polynomials.
The Hamiltonian with the P\"oschl-Teller (PT) potential is given by
\begin{align}
\mathcal{H}=-\frac{d^2}{dx^2}+U(x;g,h), \quad 
U(x;g,h)=  {\displaystyle \frac{g(g-1)}{\sin^2x}+\frac{h(h-1)}{\cos^2x}-(g+h)^2}.
\label{pot}
\end{align}
The eigenstate with the energy $\mathcal{E}_n(g,h) = 4 n(n+g+h)$ $(  n=0,1,2, \dots )$ is given by
\begin{align}
\phi_n(x;g, h)&=(\sin x)^g (\cos x)^h P^{(g - 1/2, h - 1/2)}_n(\eta(x)),\quad 
\eta(x) = \cos(2x) ,
\end{align}
where $P_n^{(\alpha,\beta )}(\eta )$ is the Jacobi polynomial in the variable $\eta$ defined by
\begin{equation}
P _n^{(\alpha , \beta )}(\eta )= \frac{(\alpha +1)_n}{n!} \sum _{k=0}^n \frac{(-n)_k(n+\alpha +\beta +1)_k}{k! (\alpha +1) _k} \left( \frac{1-\eta }{2} \right)^k .
\end{equation}
Then the eigenstate is square-integrable $\int _0 ^{\pi /2} \phi_n(x;g, h)^2 dx < +\infty $ in the case $g>-1/2$, $h>-1/2$.
To define the multi-indexed Jacobi polynomials (\cite{os25,GKM}), we introduce three types of seed polynomial solutions indexed by $\text{v}\in\mathbb{Z}_{\geq 0}$:
\begin{align}
    \tilde{\phi}_\text{v}^\ai(x;g,h)&= (\sin x)^g (\cos x)^{1 - h} P^{(g - 1/2, 1/2-h)}_\text{v}(\eta(x)),
 \label{seed1}\\
  \tilde{\phi}_\text{v}^\ait(x;g,h)&= (\sin x)^{1-g} (\cos x)^h P^{(1/2-g, h-1/2)}_\text{v}(\eta(x)), \nonumber \\
 \tilde{\phi}_\text{v}^\aitr (x;g,h)&=   (\sin x)^{1-g} (\cos x)^{1 - h}  P^{(1/2-g, 1/2-h)}_\text{v}(\eta(x)). \nonumber
\end{align}
They are solutions of the Schr\"odinger equation \eqref{pot} with the eigenvalues $\tilde{\mathcal{E}}_\text{v}^\ai(g,h)=-4 (g + \text{v} + 1/2) (h   -\text{v}-1/2) $, $\tilde{\mathcal{E}}_\text{v}^\ait(g,h)= \tilde{\mathcal{E}}_{-(\text{v}+1)}^\ai (g,h) $, $\tilde{\mathcal{E}}_\text{v}^\aitr(g,h) =\mathcal{E}_{-(\text{v}+1)}(g,h) $ respectively, which are not square-integrable in the case $g\geq 3/2$, $h \geq 3/2$.
In this paper we assume that $g \pm h \not \in \Z $ and $g,h \not \in \Z +1/2 $ under which the distinct eigenstates and seed solutions are linearly independent.
Let $\varphi _j $ be a seed solution or an eigenstate for $j=1,\dots ,{\mathcal N}$ and assume that $\varphi _1, \dots , \varphi _{{\mathcal N}}$ are distinct.
Let $W[ \varphi _1, \dots , \varphi _{{\mathcal N'}} ](x)$ be the Wronskian with respect to the derivative of the variable $x$.
Then it follows from the typical argument (\cite{crum,adler,Krein}) that 
\begin{equation}
\phi^{(\mathcal{N})}_n(x)=
\frac{\text{W}[\varphi _1, \dots , \varphi _{{\mathcal N}},\phi_n](x)}{\text{W}[\varphi _1, \dots , \varphi _{{\mathcal N}}](x)}
\label{eq:phiMn}
\end{equation}
 is an eigenfunction of the deformed Hamiltonian 
\begin{align}
&{\mathcal H}^{(\mathcal{N})}=-\frac{d^2}{dx^2}+ U(x;g,h)-2\frac{d^2\log\text{W}[\varphi _1, \dots , \varphi _{{\mathcal N}}](x)}{dx^2}
\end{align}
with the same eigenvalue $\mathcal{E}_n=4n(n+g+h)$, provided that the deformed potential is non-singular on the open interval $(0,\pi /2)$, $\varphi _1, \dots , \varphi _{{\mathcal N}},\phi_n$ are distinct, and $g$, $h$ are enough large (see also \cite{os25}).
The multi-indexed Jacobi polynomial is defined by the polynomial part of the denominator of $\phi^{(\mathcal{N})}_n $.
In this paper we extend the notion of the multi-indexed Jacobi polynomial such that the polynomial part of the Wronskian $\text{W}[\varphi _1, \dots , \varphi _{{\mathcal N}},\phi_n](x) $ or $\text{W}[\varphi _1, \dots , \varphi _{{\mathcal N}}](x) $ in the variable $\eta $.
For example the Wronskian $\text{W}[\tilde{\phi } ^{\ai }_{1} \tilde{\phi } ^{\ait }_{2} \tilde{\phi } ^{\aitr }_{1} ](x)$ is equal to $\frac{(2g-1)(2h+1)}{16} (\sin x) ^{1-g} (\cos x)^{1-h } P(\eta (x) )$, where $P(\eta )$ is a multi-indexed Jacobi polynomial of degree $5$ such that the coefficient of $\eta ^5 $ is $ (g - h+2) (g - h -1) ( g - h-3 ) ( g - h-4 ) ( g + h -3) $.
We remark that the deformed potentials may coincide for a different choice of seed solutions  $\varphi _1, \dots , \varphi _{{\mathcal N}}$ and  $\varphi '_1, \dots , \varphi '_{{\mathcal N}'}$.

In this paper, we connect the tuple of seed solutions $\varphi _1, \dots , \varphi _{{\mathcal N}}$ to a pair of Maya diagram with a division, and we show that the Maya diagrams describe relations among the Wronskians $\text{W}[\varphi _1, \dots , \varphi _{{\mathcal N}}](x)$.
As a corollary, the polynomial part of the Wronskian $\text{W}[\varphi _1, \dots , \varphi _{{\mathcal N}}](x)$ is proportional to the polynomial part of  some Wronskian which constitutes the type $\ai $ seed solutions and the square-integrable eigenstates with shifted parameters.
Let us explain our results by an example.
The tuple $\tilde{\phi } ^{\ai }_{1} \tilde{\phi } ^{\ait }_{2} \tilde{\phi } ^{\aitr }_{1} $ corresponds to 
$$ 
(\aitr ) \:  {{\dots \bullet \bullet \circ \bullet | \circ  \circ \circ \circ \dots  }\atop{\dots  3\: 2\: 1\: 0\; \; 0\: 1\: 2\: 3 \dots  }} \: (\mbox{eigenstates}) , \; (\ait) \: {{\dots \bullet \bullet \circ \bullet \bullet | \circ \bullet \circ \circ \dots  }\atop{\dots 4 \: 3 \: 2\: 1\: 0 \; \; 0\: 1\: 2\: 3 \dots  }}\: (\ai) , 
$$ 
where the white (resp. black) beads in the left (resp. right) of the division of the first Maya diagram represented  the type $\aitr $ seed solutions (resp. the eigenstates) and the white (resp. black) beads in the left (resp. right) of the division of the second Maya diagram represent  the type $\ait $ (resp. type $\ai $) seed solutions.
We move the division of the second Maya diagram one step to the left. Then the resulting Maya diagrams with divisions are 
$$ 
{{\dots \bullet \bullet \circ \bullet | \circ  \circ \circ \circ \dots  }\atop{\dots  3\: 2\: 1\: 0\; \; 0\: 1\: 2\: 3 \dots  }}, \; {{\dots \bullet \bullet \circ \bullet | \bullet \circ \bullet \circ \circ \dots  }\atop{\dots 3 \: 2\: 1\: 0 \; \; 0\: 1\: 2\: 3 \: 4 \dots  }}, 
$$ 
the corresponding tuple of states is  $\tilde{\phi } ^{\ai }_{0} \tilde{\phi } ^{\ai }_{2}  \tilde{\phi } ^{\ait }_{1} \tilde{\phi } ^{\aitr }_{1} $ and the relation between the Wronskians is 
$ W [\tilde{\phi } ^{\ai }_{1} \tilde{\phi } ^{\ait }_{2} \tilde{\phi } ^{\aitr }_{1} ](x;g,h ) \propto W[ \tilde{\phi } ^{\ai }_{0} \tilde{\phi } ^{\ai }_{2}  \tilde{\phi } ^{\ait }_{1} \tilde{\phi } ^{\aitr }_{1} ] (x;g-1,h+1) (\sin x)^{1-g} (\cos x)^{h} $.
Note that similar formulas were obtained in \cite{os25,osMW,Odake}.
We move the division of the first and the second Maya diagrams repeatedly to the left.
Then we have the Maya diagrams with divisions 
$$ 
{{\dots \bullet \bullet | \circ \bullet \circ  \circ \circ \circ \dots  }\atop{\dots 1\: 0\; \; 0\: 1\: 2\: 3 \: 4 \: 5 \dots  }}, \; {{\dots \bullet \bullet | \circ \bullet \bullet \circ \bullet \circ \circ \dots  }\atop{\dots 1\: 0 \; \; 0\: 1\: 2\: 3\: 4 \: 5 \: 6 \dots  }}, 
$$ 
and the relation
$ W [\tilde{\phi } ^{\ai }_{1} \tilde{\phi } ^{\ait }_{2} \tilde{\phi } ^{\aitr }_{1} ](x;g,h ) \propto W[ \tilde{\phi } ^{\ai }_{1} \tilde{\phi } ^{\ai }_{2}  \tilde{\phi } ^{\ai }_{4} \phi _{1} ](x,g-5,h+1) (\sin x)^{15-5g} (\cos x)^{h} $.
Thus the given Wronskian is proportional to the Wronskian which consists of the first seed solutions and the eigenstates.
For the general statement, see the propositions and the theorem in section \ref{sec:MW}.
Note that Odake \cite{Odake} established that the polynomial part of the Wronskian $\text{W}[\varphi _1, \dots , \varphi _{{\mathcal N}}](x)$ where $\varphi _i $ is a type $\ai $ seed solution or a type $\ait $ seed solution is proportional to the polynomial part of some Wronskian which constitutes the type $\ai $ seed solutions with shifted parameters, where the systems include the discrete quantum mechanics.

We give an application to the eigenstates of the deformed PT system which are represented by deleting a seed solution in section \ref{sec:Apee}.
In an example of the system governed by the Hamiltonian  
\begin{align}
&{\mathcal H}^{(3)}=-\frac{d^2}{dx^2}+ U(x;g,h)-2\frac{d^2\log\text{W}[\tilde{\phi } ^{\ai }_{1} \tilde{\phi } ^{\ait }_{2} \tilde{\phi } ^{\aitr }_{1} ](x)}{dx^2} ,
\end{align}
the function 
\begin{equation}
\phi^{(3)}_{-2}(x)=
\frac{\text{W}[\tilde{\phi } ^{\ai }_{1} \tilde{\phi } ^{\ait }_{2} ](x)}{\text{W}[\tilde{\phi } ^{\ai }_{1} \tilde{\phi } ^{\ait }_{2} \tilde{\phi } ^{\aitr }_{1} ](x)}
\label{eq:phi3ndel}
\end{equation}
is an eigenfunction of the Hamiltonian with the eigenvalue $ \mathcal{E}_{-2}=-8 (g+h-2)$.
As a consequence, the positions of the white beads of the first Maya diagram describe the energies of the system.
Note that the eigenfunctions represented by deleting a seed solution were considered in \cite{osKA,GGM}
in the different situations.

This article is organized as follows.
In section \ref{sec:DT}, we obtain formulas which are used  later.
Section \ref{sec:MW} is the main part of this paper.
In section \ref{sec:Apee}, we give an application to the  extra eigenstates of the deformed PT system.
In section \ref{sec:rem}, we give concluding remarks.

\section{Darboux transformation and relations of Wronskian} \label{sec:DT}

We apply Darboux transformation with respect to the seed solution $\tilde{\phi}_\text{0}^\ai(x;g,h) =  (\sin x)^g (\cos x)^{1 - h}  $ in the PT system.
Then 
\begin{align}
\phi^{(1)}_n(x)= \frac{\text{W}[ \tilde{\phi}_\text{0}^\ai,\phi_n](x;g,h)}{\tilde{\phi}_\text{0}^\ai(x;g,h)}, \quad 
\tilde{\phi }^{(1), \ai }_\text{v}(x)= \frac{\text{W}[ \tilde{\phi}_\text{0}^\ai, \tilde{\phi}_\text{v} ^\ai ](x;g,h)}{\tilde{\phi}_\text{0}^\ai(x;g,h)}, 
\label{eq:phi2} \\
 \tilde{\phi }^{(1), \ait }_\text{v}(x)= \frac{\text{W}[ \tilde{\phi}_\text{0}^\ai, \tilde{\phi}_\text{v} ^\ait ](x;g,h)}{\tilde{\phi}_\text{0}^\ai(x;g,h)},  \quad 
\tilde{\phi }^{(1), \aitr }_\text{v}(x)= \frac{\text{W}[ \tilde{\phi}_\text{0}^\ai, \tilde{\phi}_\text{v} ^\aitr ](x;g,h)}{\tilde{\phi}_\text{0}^\ai(x;g,h)}, \nonumber
\end{align}
are eigenfunctions of the deformed Hamiltonian 
\begin{align}
&{\mathcal H}^{(1)}=-\frac{d^2}{dx^2}+ U(x;g,h)-2\frac{d^2\log ((\sin x)^g (\cos x)^{1 - h}  )}{dx^2}=-\frac{d^2}{dx^2}+ U(x;g+1,h-1)
\end{align}
with the eigenvalues $\mathcal{E}_n (g,h) =4n(n+g+h)$, $\tilde{\mathcal{E}}_\text{v}^\ai (g,h)=-4 (g + \text{v} + 1/2) (h   -\text{v}-1/2) $, $\tilde{\mathcal{E}}_\text{v}^\ait(g,h)= -4 (g - \text{v} - 1/2) (h  +\text{v} + 1/2)  $, $\tilde{\mathcal{E}}_\text{v}^\aitr(g,h) =-4( \text{v} + 1) (g+h  -\text{v} - 1)$ respectively.
On the other hand, the functions $\phi_n (x;g+1,h-1) $, $\tilde{\phi}_{\text{v}-1} ^\ai (x;g+1,h-1) $, $\tilde{\phi}_{\text{v} +1} ^\ait (x;g+1,h-1) $, $\tilde{\phi}_\text{v} ^\aitr (x;g+1,h-1) $ are also an eigenfunction of $-\frac{d^2}{dx^2}+ U(x;g+1,h-1) $ with the eigenvalues $4n(n+g+h)$, $-4 (g + \text{v} + 1/2) (h   -\text{v}-1/2) $, $ -4 (g - \text{v} - 1/2) (h  +\text{v} + 1/2) $, $-4( \text{v} + 1) (g+h  -\text{v} - 1)$ respectively.
By comparing eigenfunctions with the same eigenvalue, we have the following relations;
\begin{eqnarray}
& & \text{W}[\tilde{\phi } ^{\ai }_{0},\phi_n](x;g,h) \propto \phi_n (x;g +1,h-1) \tilde{\phi } ^{\ai }_{0}(x;g,h) , \label{eq:WI} \\
& & \text{W}[\tilde{\phi } ^{\ai }_{0}, \tilde{\phi } ^{\ai }_n](x;g,h) \propto \tilde{\phi } ^{\ai }_{n-1} (x;g +1,h-1) \tilde{\phi } ^{\ai }_{0}(x;g,h) , \nonumber \\
& & \text{W}[\tilde{\phi } ^{\ai }_{0}, \tilde{\phi } ^{\ait }_n](x;g,h) \propto \tilde{\phi } ^{\ait }_{n +1} (x;g+1,h-1) \tilde{\phi } ^{\ai }_{0} (x;g,h) , \nonumber \\
& &  \text{W}[\tilde{\phi } ^{\ai }_{0}, \tilde{\phi } ^{\aitr }_n ](x;g,h) \propto \tilde{\phi } ^{\aitr }_n (x;g+1,h-1) \tilde{\phi } ^{\ai }_{0}(x;g,h) . \nonumber 
\end{eqnarray}
\begin{proof}
The functions $\phi^{(1)}_n(x)$ and $\phi_n (x;g+1,h-1) $ satisfy the linear differential equation $\{ -\frac{d^2}{dx^2}+ U(x;g+1,h-1) - 4n(n+g+h) \} f(x)=0$.
Thus they belong to the two dimensional vector space of solutions of the differential equation.
The differential equation has a regular singularity along $x=0$ and the exponents are $g+1$ and $-g$, and any solutions of the differential equation is written as a liner combination of $x^{g+1} f_1(x^2)$ and $x^{-g} f_2 (x^2)$, where $f_1(t)$ and $f_2(t)$ are convergent power series of $t$ such that $f_1(1)=f_2(1)=1$.  
The function $\phi_n (x;g+1,h-1)  $ is expanded as $x^{g+1} (c_0 +c _2 x^2 +c_4 x^4+ \dots ) $.
On the other hand, the function $\phi^{(1)}_n(x) = \phi ' _n (x;g,h) - \phi  _n (x;g,h) \tilde{\phi}_\text{0}^\ai(x;g,h) '  / \tilde{\phi}_\text{0}^\ai(x;g,h)$ are  expanded as $x^{g+1} (c'_0 +c' _2 x^2 +c'_4 x^4+ \dots ) $.
Therefore the function $\phi^{(1)}_n(x)$ is proportional to the function $\phi_n (x;g+1,h-1)  $, i.e.  $ \text{W}[\tilde{\phi } ^{\ai }_{0},\phi_n](x;g,h) \propto \phi_n (x;g +1,h-1) \tilde{\phi } ^{\ai }_{0}(x;g,h) $.
Other relations are shown similarly.
\end{proof}
We denote the square-integrable eigenstate by the type "null" eigenfunction, i.e. $\tilde{\phi } ^{\nul }_{n} (x;g,h)=  \phi  _{n} (x;g,h)$.
Then the following proposition is proved by similar ways to obtain Eq.(\ref{eq:WI}).
\begin{prop}
Let $ J \in \{ \ai ,\ait , \aitr , \nul \}$.
We have
\begin{align}
& \text{W}[\tilde{\phi } ^{\ai }_{0}, \tilde{\phi } ^{J }_{n} ](x;g,h) \propto \tilde{\phi } ^{J }_{n -(\ai ,J) }  (x;g+1 ,h-1 ) \tilde{\phi } ^{\ai }_{0} (x;g,h) , \label{eq:W0Jn} \\
& \text{W}[\tilde{\phi } ^{\ait }_{0}, \tilde{\phi } ^{J }_{n} ](x;g,h) \propto \tilde{\phi } ^{J }_{n -(\ait ,J) }  (x;g-1 ,h+1 ) \tilde{\phi } ^{\ait }_{0} (x;g,h) , \nonumber \\
& \text{W}[\tilde{\phi } ^{\aitr }_{0}, \tilde{\phi } ^{J }_{n} ](x;g,h) \propto \tilde{\phi } ^{J }_{n -(\aitr ,J) }  (x;g-1 ,h-1 ) \tilde{\phi } ^{\aitr }_{0} (x;g,h) , \nonumber \\
& \text{W}[{\phi } _{0}, \tilde{\phi } ^{J }_{n} ](x;g,h) \propto \tilde{\phi } ^{J }_{n -(N ,J) }  (x;g+1 ,h+1) {\phi } _{0} (x;g,h) , \nonumber 
\end{align}
where 
\begin{equation}
(J,J')= \left\{ 
\begin{array}{ll}
1 & J=J',\\
-1 & \{ J, J' \} =\{ \ai ,\ait \}  \mbox{ or } \{\aitr , \nul \} ,\\
0 & \mbox{otherwise.}
\end{array}
\right.
\end{equation}
\end{prop}
By applying the formula of Wronskians
\begin{equation}
\text{W}[\varphi _1, \dots , \varphi _M ,f,g ](x) \text{W}[\varphi _1, \dots , \varphi _M ](x) =  \text{W}[ \text{W}[\varphi _1, \dots , \varphi _M ,f](x) , \text{W}[\varphi _1, \dots , \varphi _M ,g](x)] 
\end{equation}
(see \cite{crum}), we have the following relations;
\begin{prop} \label{prop:ai0M}
Let $ t_j \in \{ \ai ,\ait , \aitr , \nul \}$ and $n_j \in \{ 0,1,2, \dots \}$.
We have
\begin{align}
& \text{W}[\tilde{\phi } ^{\ai }_{0}, \tilde{\phi } ^{t_1 }_{n_1}  , \ldots ,\tilde{\phi } ^{t_M }_{n_M} ](x;g,h)  \propto \text{W}[  \tilde{\phi } ^{t_1 }_{n_1 -(\ai ,t_1) }  , \ldots ,\tilde{\phi } ^{t_M }_{n_M -(\ai ,t_M) } ] (x;g+1 ,h-1 ) \tilde{\phi } ^{\ai }_{0} (x;g,h) , \\
& \text{W}[\tilde{\phi } ^{\ait }_{0}, \tilde{\phi } ^{t_1 }_{n_1}  , \ldots ,\tilde{\phi } ^{t_M }_{n_M} ](x;g,h)  \propto \text{W}[  \tilde{\phi } ^{t_1 }_{n_1 -(\ait ,t_1) }  , \ldots ,\tilde{\phi } ^{t_M }_{n_M -(\ait ,t_M) } ] (x;g-1 ,h+1 ) \tilde{\phi } ^{\ait }_{0} (x;g,h) , \nonumber \\
& \text{W}[\tilde{\phi } ^{\aitr }_{0}, \tilde{\phi } ^{t_1 }_{n_1}  , \ldots ,\tilde{\phi } ^{t_M }_{n_M} ](x;g,h)  \propto \text{W}[  \tilde{\phi } ^{t_1 }_{n_1 -(\aitr ,t_1) }  , \ldots ,\tilde{\phi } ^{t_M }_{n_M -(\aitr ,t_M) } ] (x;g-1 ,h-1 ) \tilde{\phi } ^{\aitr }_{0} (x;g,h) , \nonumber \\
& \text{W}[{\phi }_{0}, \tilde{\phi } ^{t_1 }_{n_1}  , \ldots ,\tilde{\phi } ^{t_M }_{n_M} ](x;g,h)  \propto \text{W}[  \tilde{\phi } ^{t_1 }_{n_1 -(\nul  ,t_1) }  , \ldots ,\tilde{\phi } ^{t_M }_{n_M -(N ,t_M) } ] (x;g+1 ,h+1 ) {\phi } _{0} (x;g,h). \nonumber 
\end{align}
\end{prop}
\begin{proof}
We show the first relation by the induction of $M$, because the others are shown similarly.
The case $M=1$ is true.
We assume that the case $M\leq k$ is true.
Then 
\begin{align*}
& \text{W}[\tilde{\phi } ^{\ai }_{0}, \tilde{\phi } ^{t_1 }_{n_1}  , \ldots  ,\tilde{\phi } ^{t_{k-1} }_{n_{k-1}}  ,\tilde{\phi } ^{t_k }_{n_k}  ,\tilde{\phi } ^{t_{k+1} }_{n_{k+1}} ](x;g,h)  \text{W}[\tilde{\phi } ^{\ai }_{0}, \tilde{\phi } ^{t_1 }_{n_1}  , \ldots  ,\tilde{\phi } ^{t_{k-1} }_{n_{k-1}}  ](x;g,h)  \\
& = \text{W}[ \text{W}[\tilde{\phi } ^{\ai }_{0}, \tilde{\phi } ^{t_1 }_{n_1}  , \ldots  ,\tilde{\phi } ^{t_{k-1} }_{n_{k-1}}  ,\tilde{\phi } ^{t_k }_{n_k}  ](x;g,h) , \text{W}[\tilde{\phi } ^{\ai }_{0}, \tilde{\phi } ^{t_1 }_{n_1}  , \ldots  ,\tilde{\phi } ^{t_{k-1} }_{n_{k-1}}   ,\tilde{\phi } ^{t_{k+1} }_{n_{k+1}} ](x;g,h)  ] \\
& \propto \text{W}[ \text{W}[\tilde{\phi } ^{t_1 }_{n_1 -(\ai ,t_1) }  , \ldots  ,\tilde{\phi } ^{t_{k-1} }_{n_{k-1}-(\ai ,t_{k-1}) }  ,\tilde{\phi } ^{t_k }_{n_k -(\ai ,t_k) } ](x;g+1,h-1) \tilde{\phi } ^{\ai }_{0} (x;g,h) , \\
& \qquad \text{W}[ \tilde{\phi } ^{t_1 }_{n_1-(\ai ,t_{1})}  , \ldots  ,\tilde{\phi } ^{t_{k-1} }_{n_{k-1}-(\ai ,t_{k-1})}  ,\tilde{\phi } ^{t_{k+1} }_{n_{k+1}-(\ai ,t_{k+1})} ](x;g+1,h-1) \tilde{\phi } ^{\ai }_{0} (x;g,h)  ] \\
& =  \text{W}[ \text{W}[\tilde{\phi } ^{t_1 }_{n_1 -(\ai ,t_1) }  , \ldots  ,\tilde{\phi } ^{t_{k-1} }_{n_{k-1}-(\ai ,t_{k-1}) }  ,\tilde{\phi } ^{t_k }_{n_k -(\ai ,t_k) } ](x;g+1,h-1)  , \\
&  \qquad \text{W}[ \tilde{\phi } ^{t_1 }_{n_1-(\ai ,t_{1})}  , \ldots  ,\tilde{\phi } ^{t_{k-1} }_{n_{k-1}-(\ai ,t_{k-1})}  ,\tilde{\phi } ^{t_{k+1} }_{n_{k+1}-(\ai ,t_{k+1})}   ](x;g+1,h-1)  ] \tilde{\phi } ^{\ai }_{0} (x;g,h) ^2 \\
& =  \text{W}[\tilde{\phi } ^{t_1 }_{n_1 -(\ai ,t_1) }  , \ldots  ,\tilde{\phi } ^{t_{k-1} }_{n_{k-1}-(\ai ,t_{k-1}) }  ,\tilde{\phi } ^{t_k }_{n_k -(\ai ,t_k) } ,\tilde{\phi } ^{t_{k+1} }_{n_{k+1}-(\ai ,t_{k+1})} ](x;g+1,h-1)  \\
&  \qquad \text{W}[ \tilde{\phi } ^{t_1 }_{n_1-(\ai ,t_{1})}  , \ldots  ,\tilde{\phi } ^{t_{k-1} }_{n_{k-1}-(\ai ,t_{k-1})}    ](x;g+1,h-1)  ] \tilde{\phi } ^{\ai }_{0} (x;g,h) ^2 .
\end{align*}
Combining with 
\begin{align*}
&   \text{W}[\tilde{\phi } ^{\ai }_{0}, \tilde{\phi } ^{t_1 }_{n_1}  , \ldots  ,\tilde{\phi } ^{t_{k-1} }_{n_{k-1}}  ](x;g,h)  \\
& \propto  \text{W}[ \tilde{\phi } ^{t_1 }_{n_1-(\ai ,t_{1})}  , \ldots  ,\tilde{\phi } ^{t_{k-1} }_{n_{k-1}-(\ai ,t_{k-1})}    ](x;g+1,h-1)   \tilde{\phi } ^{\ai }_{0} (x;g,h) ,
\end{align*}
we obtain the proposition for the case $M=k+1$.
\end{proof}

\section{Maya diagrams and Wronskian} \label{sec:MW}

A Maya diagram is a placement of white beads ($\circ $) or black ones ($\bullet  $) in a line such that the beads at a sufficiently negative (left) position are black and the ones at a sufficiently positive (right) position are white.

For the tuple of states 
\begin{equation}
\tilde{\phi } ^{\ai }_{d^{\ai }_1},\tilde{\phi } ^{\ai }_{d^{\ai }_2},\ldots,\tilde{\phi } ^{\ai }_{d^{\ai }_{M_{\ai }}}, \tilde{\phi } ^{\ait }_{d^{\ait }_1},\ldots,\tilde{\phi } ^{\ait }_{d^{\ait }_{M_{\ait }}},   \tilde{\phi } ^{\aitr }_{d^{\aitr }_1},\ldots,\tilde{\phi } ^{\aitr }_{d^{\aitr }_{M_{\aitr }}}, \phi _{d^{\nul }_1},\ldots, \phi _{d^{\nul }_{M_{\nul }}}, \label{eq:tuple}
\end{equation}
we associate a couple of Maya diagrams with a division as follows;\\
1. The $d^{\aitr }_k +1 $-st bead to the left of the division of the first Maya diagram is white ($\circ $) for $k=1,2,\dots , M_{\aitr }$, and the remainder of the beads to the left of the division is black $(\bullet )$.\\
2. The $d^{\nul }_k +1 $-st bead to the right of the division of the first Maya diagram is black ($\bullet $) for $k=1,2,\dots , M_{\nul }$, and the remainder of the beads to the right of the division is white $(\circ )$.\\
3. The $d^{\ait }_k +1 $-st bead to the right of the division of the second Maya diagram is black ($\bullet $) for $k=1,2,\dots , M_{\ait }$, and the remainder of the beads to the right of the division is white $(\circ )$.\\
4. The $d^{\ai }_k +1 $-st bead to the left of the division of the second Maya diagram is white ($\circ $) for $k=1,2,\dots , M_{\ai }$, and the remainder of the beads to the left of the division is black $(\bullet )$.\\
An example is given in the introduction.
We give another example here.
The tuple of states $\tilde{\phi } ^{\ai }_{2} \tilde{\phi } ^{\ai }_{3} \tilde{\phi } ^{\ait }_{0} \tilde{\phi } ^{\ait }_{2} \tilde{\phi } ^{\aitr }_{3} \phi _{0} \phi _{1} $ corresponds to the following Maya diagrams; 
$$ \dots \bullet  \bullet \bullet \circ \bullet \bullet \bullet | \bullet \bullet \circ \circ \circ \dots , \quad \dots \bullet \bullet \bullet \circ \bullet \circ | \circ \circ \bullet \bullet \circ \circ \circ \dots .$$
We go back to the general situation where the tuple of state is given by Eq.(\ref{eq:tuple}) and assume that $ 0 \leq d^{\text J}_1 < d^{\text J}_2 < \dots < d^{\text J}_{M_{\text J }}$ for $J \in \{ \ai ,\ait , \aitr , \nul \} $.
We move the division of the second Maya diagram one step to the left.
Then the corresponding tuple of states is 
\begin{align}
&  \tilde{\phi } ^{\ai }_{0}, \tilde{\phi } ^{\ai }_{d^{\ai }_1 +1},\ldots,\tilde{\phi } ^{\ai }_{d^{\ai }_{M_{\ai }}+1}, \tilde{\phi } ^{\ait }_{d^{\ait }_1 -1},\ldots,\tilde{\phi } ^{\ait }_{d^{\ait }_{M_{\ait }}-1}, \tilde{\phi } ^{\aitr }_{d^{\aitr }_1},\ldots,\tilde{\phi } ^{\aitr }_{d^{\aitr }_{M_{\aitr }}}, \phi _{d^{\nul }_1},\ldots, \phi _{d^{\nul }_{M_{\nul }}} , & d^{\ait }_1 \neq 0 ,\\
& \tilde{\phi } ^{\ai }_{d^{\ai }_1 +1 },\ldots,\tilde{\phi } ^{\ai }_{d^{\ai }_{M_{\ai }}+1}, \tilde{\phi } ^{\ait }_{d^{\ait }_2 -1},\ldots,\tilde{\phi } ^{\ait }_{d^{\ait }_{M_{\ait }}-1},  \tilde{\phi } ^{\aitr }_{d^{\aitr }_1},\ldots,\tilde{\phi } ^{\aitr }_{d^{\aitr }_{M_{\aitr }}}, \phi _{d^{\nul }_1},\ldots, \phi _{d^{\nul }_{M_{\nul }}} ,& d^{\ait }_1 = 0 . \nonumber
\end{align}
In the example of $\tilde{\phi } ^{\ai }_{2} \tilde{\phi } ^{\ai }_{3} \tilde{\phi } ^{\ait }_{0} \tilde{\phi } ^{\ait }_{2} \tilde{\phi } ^{\aitr }_{3} \phi _{0} \phi _{1} $, we have $ d^{\ait }_1 = 0 $, the changed Maya diagrams are 
$$ \dots \bullet  \bullet \bullet \circ \bullet \bullet \bullet | \bullet \bullet \circ \circ \circ \dots , \quad \dots \bullet \bullet \bullet \circ \bullet | \circ \circ \circ \bullet \bullet \circ \circ \circ \dots .$$
and the tuple of state corresponding to the changed Maya diagrams is $\tilde{\phi } ^{\ai }_{3} \tilde{\phi } ^{\ai }_{4} \tilde{\phi } ^{\ait }_{1} \tilde{\phi } ^{\aitr }_{3} \phi _{0} \phi _{1}  $.

It follows from Proposition \ref{prop:ai0M} that the Wronskian
\begin{equation}
\text{W}[\tilde{\phi } ^{\ai }_{d^{\ai }_1},\ldots,\tilde{\phi } ^{\ai }_{d^{\ai }_{M_{\ai }}}, \tilde{\phi } ^{\ait }_{d^{\ait }_1},\ldots,\tilde{\phi } ^{\ait }_{d^{\ait }_{M_{\ait }}},   \tilde{\phi } ^{\aitr }_{d^{\aitr }_1},\ldots,\tilde{\phi } ^{\aitr }_{d^{\aitr }_{M_{\aitr }}}, \phi _{d^{\nul }_1}, \ldots, \phi _{d^{\nul }_{M_{\nul }}}] (x;g,h) 
\label{eq:Wron}
\end{equation}
is proportional to 
\begin{align}
& \text{W}[ \tilde{\phi } ^{\ai }_{0}, \tilde{\phi } ^{\ai }_{d^{\ai }_1 +1},\ldots,\tilde{\phi } ^{\ai }_{d^{\ai }_{M_{\ai }}+1}, \tilde{\phi } ^{\ait }_{d^{\ait }_1 -1},\ldots,\tilde{\phi } ^{\ait }_{d^{\ait }_{M_{\ait }}-1}, \tilde{\phi } ^{\aitr }_{d^{\aitr }_1},\ldots,\tilde{\phi } ^{\aitr }_{d^{\aitr }_{M_{\aitr }}}, \phi _{d^{\nul }_1},\ldots \\
& \qquad \ldots, \phi _{d^{\nul }_{M_{\nul }}} ]  (x,g-1,h+1)  / \tilde{\phi } ^{\ai }_{0} (x;g-1,h+1) \nonumber  
\end{align}
for $d^{\ait }_1 \neq 0  $ and to 
\begin{align}
& \text{W}[\tilde{\phi } ^{\ai }_{d^{\ai }_1 +1 },\ldots,\tilde{\phi } ^{\ai }_{d^{\ai }_{M_{\ai }}+1}, \tilde{\phi } ^{\ait }_{d^{\ait }_2 -1},\ldots,\tilde{\phi } ^{\ait }_{d^{\ait }_{M_{\ait }}-1},  \tilde{\phi } ^{\aitr }_{d^{\aitr }_1},\ldots,\tilde{\phi } ^{\aitr }_{d^{\aitr }_{M_{\aitr }}}, \phi _{d^{\nul }_1},\ldots \\
& \qquad \ldots , \phi _{d^{\nul }_{M_{\nul }}} ] (x,g-1,h+1)  \tilde{\phi } ^{\ait }_{0} (x;g,h) \nonumber  
\end{align}
for $d^{\ait }_1 =0 $.
Note that in the case that the states in the Wronskian are type $ai$ and $\ait$, the relations were obtained by Odake and Sasaki \cite{os25,osMW}, although the proof seems to be different from ours. 
By using $1/ \tilde{\phi } ^{\ai }_{0} (x;g-1,h+1)  = \tilde{\phi } ^{\ait }_{0} (x;g,h)  = (\sin x)^{1-g} (\cos x)^{h}  $, we have the following proposition on the movement of the division of the second Maya diagram.
\begin{prop}
If the division of the second Maya diagram is moved one step to the left, then the original Wronskian is proportional to the product of the function $ (\sin x)^{1-g} (\cos x)^{h} $ and the Wronskian corresponding to the moved division where the parameters are shifted to $g-1$ and $h+1$.
\end{prop}
On the movement of the division of the first Maya diagram, we have
\begin{prop}
If the division of the first Maya diagram is moved one step to the left, then the original Wronskian is proportional to the product of the function $ (\sin x)^{1-g} (\cos x)^{1-h} $ and the Wronskian corresponding to the moved division where the parameters are shifted to $g-1$ and $h-1$, i.e. the Wronskian given as Eq.(\ref{eq:Wron}) is proportional to
\begin{align}
& \text{W}[  \tilde{\phi } ^{\ai }_{d^{\ai }_1 },\ldots,\tilde{\phi } ^{\ai }_{d^{\ai }_{M_{\ai }}}, \tilde{\phi } ^{\ait }_{d^{\ait }_1 },\ldots,\tilde{\phi } ^{\ait }_{d^{\ait }_{M_{\ait }}}, \tilde{\phi } ^{\aitr }_{d^{\aitr }_1 -1},\ldots,\tilde{\phi } ^{\aitr }_{d^{\aitr }_{M_{\aitr }}-1}, \phi _{0}, \phi _{d^{\nul }_1 +1}, \ldots \\
& \qquad \ldots, \phi _{d^{\nul }_{M_{\nul }}+1} ]  (x,g-1,h-1)  (\sin x)^{1-g} (\cos x)^{1-h} \nonumber
\end{align}
for $d^{\aitr }_1 \neq 0  $ and to 
\begin{align*}
& \text{W}[\tilde{\phi } ^{\ai }_{d^{\ai }_1 },\ldots,\tilde{\phi } ^{\ai }_{d^{\ai }_{M_{\ai }}}, \tilde{\phi } ^{\ait }_{d^{\ait }_2 },\ldots,\tilde{\phi } ^{\ait }_{d^{\ait }_{M_{\ait }}},  \tilde{\phi } ^{\aitr }_{d^{\aitr }_2 -1},\ldots,\tilde{\phi } ^{\aitr }_{d^{\aitr }_{M_{\aitr }}-1}, \phi _{d^{\nul }_1 +1},\ldots \\
& \qquad \ldots , \phi _{d^{\nul }_{M_{\nul }}+1} ] (x,g-1,h-1) (\sin x)^{1-g} (\cos x)^{1-h} \nonumber
\end{align*}
for $d^{\aitr }_1 =0 $.
\end{prop}
For a given tuple of states, we associate a pair of Maya diagrams with a division and we move the divisions to the left.
By applying the propositions repeatedly, we find that the Wronskian of a given tuple of states is equal to some Wronskian which constitutes the type $\ai $ states and the square-integrable states up to a scalar multiplication.
Namely we obtain 
\begin{thm} \label{thm:red}
Let $\bar{\mathcal D}_{\text J}= \{ 0,1,2, \dots ,d^{\text J}_{M_{\text J }} \} \setminus \{ d^{\text J}_{M_{\text J }} -d^{\text J}_1  , d^{\text J}_{M_{\text J }} - d^{\text J}_2 , \dots , d^{\text J}_{M_{\text J }} - d^{\text J}_{M_{\text J }} (=0)  \}$ and write $ \bar{\mathcal D}_{\text J}= \{ e^{\text J}_1 ,  e^{\text J}_2 , \dots ,  e^{\text J}_{\bar{M}_{\text J }} \}$ ($e^{\text J}_1 <  e^{\text J}_2 < \dots < e^{\text J}_{\bar{M}_{\text J }}, \bar{M}_{\text J } = d^{\text J}_{M_{\text J }} +1- M_{\text J }$) for ${\text J } \in \{ \ai ,\ait , \aitr , \nul \}$.
We have
\begin{align}
&  \text{W}[\tilde{\phi } ^{\ai }_{d^{\ai }_1},\ldots,\tilde{\phi } ^{\ai }_{d^{\ai }_{M_{\ai }}}, \tilde{\phi } ^{\ait }_{d^{\ait }_1},\ldots,\tilde{\phi } ^{\ait }_{d^{\ait }_{M_{\ait }}},   \tilde{\phi } ^{\aitr }_{d^{\aitr }_1},\ldots,\tilde{\phi } ^{\aitr }_{d^{\aitr }_{M_{\aitr }}}, \phi _{d^{\nul }_1},\ldots, \phi _{d^{\nul }_{M_{\nul }}}] (x;g,h) \label{eq:thmainul} \\
& \propto\text{W}[\tilde{\phi } ^{\ai }_{e^{\ait }_1 },\ldots,\tilde{\phi } ^{\ai }_{e^{\ait }_{\bar{M}_{\ait }}}, \tilde{\phi } ^{\ai }_{d^{\ai }_1 +d^{\ait }_{M_{\ait }} +1},\ldots,\tilde{\phi } ^{\ai }_{d^{\ai }_{M_{\ai }}+d^{\ait }_{M_{\ait }} +1}, \tilde{\phi } ^{\nul }_{e^{\aitr }_1 },\ldots,\tilde{\phi } ^{\nul }_{e^{\aitr }_{\bar{M}_{\aitr }}},  \tilde{\phi } ^{\nul }_{d^{\nul }_1 +d^{\aitr }_{M_{\aitr }} +1},\ldots \nonumber \\
& \qquad \qquad,\tilde{\phi } ^{\nul }_{d^{\nul }_{M_{\nul }}+d^{\aitr }_{M_{\aitr }} +1} ] (x,g-d^{\ait }_{M_{\ait }}  -d^{\aitr }_{M_{\aitr }} -2 ,h+d^{\ait }_{M_{\ait }}   -d^{\aitr }_{M_{\aitr }}  ) (\sin x)^{g_{\ai, \nul}} (\cos x)^{h_{\ai, \nul}} , \nonumber
\end{align}
where $g_{\ai, \nul}= ( d^{\ait }_{M_{\ait }}  +d^{\aitr }_{M_{\aitr }} +2 ) \{ -g+(d^{\ait }_{M_{\ait }}  +d^{\aitr }_{M_{\aitr }} +3) /2 \}$ and $h_{\ai, \nul} = ( d^{\ait }_{M_{\ait }}  -d^{\aitr }_{M_{\aitr }} ) \{ h+(d^{\ait }_{M_{\ait }} - d^{\aitr }_{M_{\aitr }} -1) /2 \}$.  
\end{thm}
\begin{proof}
We move the division of the second Maya diagram of the given states $d^{\ait }_{M_{\ait }} +1  $ times to the left.
By noticing the placement of the black beads on the right of the moved division, we find that the  left-hand side of Eq.(\ref{eq:thmainul}) is proportional to
\begin{align}
& \text{W}[\tilde{\phi } ^{\ai }_{e^{\ait }_1 },\ldots,\tilde{\phi } ^{\ai }_{e^{\ait }_{\bar{M}_{\ait }}}, \tilde{\phi } ^{\ai }_{d^{\ai }_1 +d^{\ait }_{M_{\ait }} +1},\ldots,\tilde{\phi } ^{\ai }_{d^{\ai }_{M_{\ai }}+d^{\ait }_{M_{\ait }} +1},  \tilde{\phi } ^{\aitr }_{d^{\aitr }_1},\tilde{\phi } ^{\aitr }_{d^{\aitr }_2},\ldots,\tilde{\phi } ^{\aitr }_{d^{\aitr }_{M_{\aitr }}}, \\
& \qquad \qquad \phi _{d^{\nul }_1}, \phi _{d^{\nul }_2},\ldots, \phi _{d^{\nul }_{M_{\nul }}}] ] (x,g-d^{\ait }_{M_{\ait }}  -1 ,h+d^{\ait }_{M_{\ait }} +1  ) (\sin x)^{g'} (\cos x)^{h'} , \nonumber  
\end{align}
where $g' = ( d^{\ait }_{M_{\ait }}  +	1 ) \{ -g+(d^{\ait }_{M_{\ait }}  +2) /2 \}$ and $h' 
 = ( d^{\ait }_{M_{\ait }}  +1) ( h+d^{\ait }_{M_{\ait }}  /2 )$.  
We move further the division of the first Maya diagram $d^{\aitr }_{M_{\aitr }} +1  $ times to the left.
Then we obtain the theorem.
\end{proof} 
In the example of $\tilde{\phi } ^{\ai }_{2} \tilde{\phi } ^{\ai }_{3} \tilde{\phi } ^{\ait }_{0} \tilde{\phi } ^{\ait }_{2} \tilde{\phi } ^{\aitr }_{3} \phi _{0} \phi _{1} $, we have
$\bar{\mathcal D}_{\ait }= \{0,1,2 \} \setminus \{ 2,0 \} = \{ 1 \}$, $\bar{\mathcal D}_{\aitr }= \{1,2,3 \}$ and
\begin{align}
&  \text{W}[\tilde{\phi } ^{\ai }_{2} \tilde{\phi } ^{\ai }_{3} \tilde{\phi } ^{\ait }_{0} \tilde{\phi } ^{\ait }_{2} \tilde{\phi } ^{\aitr }_{3} \phi _{0} \phi _{1} ] (x;g,h) \label{eq:thmexa} \\
& \propto\text{W}[\tilde{\phi } ^{\ai }_{1} \tilde{\phi } ^{\ai }_{5} \tilde{\phi } ^{\ai }_{6} \phi _{1} \phi _{2} \phi _{3} \phi _{4} \phi _{5}] (x,g-7 ,h-1) (\sin x)^{28-7g} (\cos x)^{1-h} , \nonumber
\end{align}
and the corresponding Maya diagrams are 
$$ \dots \bullet  \bullet \bullet | \circ \bullet \bullet \bullet \bullet \bullet \circ \circ \circ \dots , \quad \dots \bullet \bullet \bullet | \circ \bullet \circ \circ \circ \bullet \bullet \circ \circ \circ \dots .$$

Note that the left-hand side of Eq.(\ref{eq:thmainul}) is also proportional to each Wronskian of the followings;
\begin{align*}
& \text{W}[\tilde{\phi } ^{\ai }_{e^{\ait }_1 },\ldots,\tilde{\phi } ^{\ai }_{e^{\ait }_{\bar{M}_{\ait }}}, \tilde{\phi } ^{\ai }_{d^{\ai }_1 +d^{\ait }_{M_{\ait }} +1},\ldots,\tilde{\phi } ^{\ai }_{d^{\ai }_{M_{\ai }}+d^{\ait }_{M_{\ait }} +1}, \tilde{\phi } ^{\aitr }_{e^{\nul }_1 },\ldots,\tilde{\phi } ^{\aitr }_{e^{\aitr }_{\bar{M}_{\nul }}}, \tilde{\phi } ^{\aitr }_{d^{\aitr }_1 +d^{\nul }_{M_{\nul }} +1},\ldots \\
& \qquad \ldots , \tilde{\phi } ^{\aitr }_{d^{\aitr }_{M_{\aitr }}+d^{\nul }_{M_{\nul }} +1} ] (x,g-d^{\ait }_{M_{\ait }}  +d^{\nul }_{M_{\nul }}  ,h+d^{\ait }_{M_{\ait }}  +d^{\nul }_{M_{\nul }} +2  ) (\sin x)^{g_{\ai, \aitr}} (\cos x)^{h_{\ai, \aitr} } , \\
& \text{W}[\tilde{\phi } ^{\ait }_{e^{\ai }_1 },\ldots,\tilde{\phi } ^{\ait }_{e^{\ai }_{\bar{M}_{\ai }}}, \tilde{\phi } ^{\ait }_{d^{\ait }_1 +d^{\ai }_{M_{\ai }} +1},\ldots,\tilde{\phi } ^{\ait }_{d^{\ait }_{M_{\ait }}+d^{\ai }_{M_{\ai }} +1},   \tilde{\phi } ^{\nul }_{e^{\aitr }_1 },\ldots,\tilde{\phi } ^{\nul }_{e^{\aitr }_{\bar{M}_{\aitr }}},  \tilde{\phi } ^{\nul }_{d^{\nul }_1 +d^{\aitr }_{M_{\aitr }} +1},\ldots \\
& \qquad \ldots ,\tilde{\phi } ^{\nul }_{d^{\nul }_{M_{\nul }}+d^{\aitr }_{M_{\aitr }} +1} ] (x,g+d^{\ai }_{M_{\ai }}  -d^{\aitr }_{M_{\aitr }} ,h-d^{\ai }_{M_{\ai }}   -d^{\aitr }_{M_{\aitr }}-2  ) (\sin x)^{g_{\ait, \nul}} (\cos x)^{h_{\ait, \nul} } ,\\
& \text{W}[\tilde{\phi } ^{\ait }_{e^{\ai }_1 },\ldots,\tilde{\phi } ^{\ait }_{e^{\ai }_{\bar{M}_{\ai }}}, \tilde{\phi } ^{\ait }_{d^{\ait }_1 +d^{\ai }_{M_{\ai }} +1},\ldots,\tilde{\phi } ^{\ait }_{d^{\ait }_{M_{\ait }}+d^{\ai }_{M_{\ai }} +1},  \tilde{\phi } ^{\aitr }_{e^{\nul }_1 },\ldots,\tilde{\phi } ^{\aitr }_{e^{\aitr }_{\bar{M}_{\nul }}},\tilde{\phi } ^{\aitr }_{d^{\aitr }_1 +d^{\nul }_{M_{\nul }} +1},\ldots \\
& \qquad \ldots ,\tilde{\phi } ^{\aitr }_{d^{\aitr }_{M_{\aitr }}+d^{\nul }_{M_{\nul }} +1} ] (x,g+d^{\ai }_{M_{\ai }}  +d^{\nul }_{M_{\nul }} +2  ,h-d^{\ai }_{M_{\ai }}  +d^{\nul }_{M_{\nul }}   ) (\sin x)^{g_{\ait, \aitr}} (\cos x)^{h_{\ait, \aitr} } ,
\end{align*}
where $g_{\ai, \aitr}= ( d^{\ait }_{M_{\ait }}  -d^{\nul }_{M_{\nul }} ) \{ g+(d^{\ait }_{M_{\ait }}  -d^{\nul }_{M_{\nul }} -1) /2 \}$, $h_{\ai, \aitr} = ( d^{\ait }_{M_{\ait }}  +d^{\nul }_{M_{\nul }} +2 ) \{ h+(d^{\ait }_{M_{\ait }} + d^{\nul }_{M_{\nul }} +1) /2 \}$, $g_{\ait, \nul}= ( d^{\ai }_{M_{\aitr }}  -d^{\aitr }_{M_{\aitr }}  ) \{ g+(d^{\ai }_{M_{\ai }}  -d^{\aitr }_{M_{\aitr }} -1 ) /2 \}$, $h_{\ait, \nul} = ( d^{\ai }_{M_{\ai }}  + d^{\aitr }_{M_{\aitr }} +2) \{ -h+(d^{\ai }_{M_{\ai }} + d^{\aitr }_{M_{\aitr }} +3) /2 \}$, $g_{\ait, \aitr}= ( d^{\nul }_{M_{\nul }}  +d^{\ai }_{M_{\ai }} +2 ) \{ g+(d^{\nul}_{M_{\nul }}  +d^{\ai }_{M_{\ai }} +1) /2 \}$ and $h_{\ait, \aitr} = ( d^{\nul }_{M_{\nul }}  -d^{\ai }_{M_{\ai }} ) \{ h+(d^{\nul }_{M_{\nul }} - d^{\ai }_{M_{\ai }} -1) /2 \}$. 

\section{Application to extra eigenstates} \label{sec:Apee}
We give an applications of the relations of the Wronskians in section \ref{sec:MW}  to the extra eigenstates of the deformed PT system.
\begin{prop} \label{prop:delstate}
Let $\varphi _j $ be a seed solution or an eigenstate for $j=1,\dots ,{\mathcal N}$ and assume that $\varphi _1, \dots , \varphi _{{\mathcal N}}$ are distinct and $\varphi _{\ell}  = \tilde{\phi } ^{\aitr }_m$.
Then 
\begin{equation}
\phi^{(\mathcal{N})}_n(x)=
\frac{\text{W}[\varphi _1, \dots , \varphi _{\ell -1}, \varphi _{\ell +1} , \dots \varphi _{{\mathcal N}}](x;g,h)}{\text{W}[\varphi _1, \dots , \varphi _{{\mathcal N}}](x;g,h)}
\label{eq:phiMndel}
\end{equation}
 is an eigenfunction of the deformed PT Hamiltonian 
\begin{align}
&{\mathcal H}^{(\mathcal{N})}=-\frac{d^2}{dx^2}+ U(x;g,h)-2\frac{d^2\log\text{W}[\varphi _1, \dots , \varphi _{{\mathcal N}}](x;g,h)}{dx^2}
\label{eq:dPTHamil}
\end{align}
with the eigenvalue $\mathcal{E}_{-m-1}=-4(m+1 ) (g+h-m-1)$, provided that the deformed potential is non-singular on the open interval $(0,\pi /2)$ and $g$, $h$ are enough large.
\end{prop}
\begin{proof}
We associate the tuple $\varphi _1, \dots , \varphi _{{\mathcal N}} $ to a pair of Maya diagrams with a division.
Then the $m +1 $-st bead to the left of the division of the first Maya diagram is white.
We move the division of the first Maya diagram $m+1$-times to the left, and denote the corresponding tuple by $\varphi '_1, \dots , \varphi '_{{\mathcal N}'}$.
Since the first bead to the right of the division is white,  no one of $\varphi '_1, \dots , \varphi '_{{\mathcal N}'} $ is proportional to $\phi _0$.
Hence the function 
\begin{equation}
\phi^{(\mathcal{N'})}_0(x)=
\frac{\text{W}[\varphi '_1, \dots \varphi '_{{\mathcal N}'},\phi _0](x;g-m-1,h-m-1)}{\text{W}[\varphi ' _1, \dots , \varphi  ' _{{\mathcal N}'}](x;g-m-1,h-m-1)}
\label{eq:phiMnp}
\end{equation}
 is an eigenfunction of the deformed Hamiltonian 
\begin{align}
&{\mathcal H}^{(\mathcal{N}')}=-\frac{d^2}{dx^2}+ V , \\
& V= U(x;g-m-1,h-m-1)-2\frac{d^2\log\text{W}[\varphi '_1, \dots , \varphi ' _{{\mathcal N}'}](x;g-m-1,h-m-1)}{dx^2} \nonumber
\end{align}
with the eigenvalue $\mathcal{E}_0=0$.
The Maya diagrams of the tuple $\varphi '_1, \dots \varphi '_{{\mathcal N}'},\phi _0 $ are obtained from the ones of the tuple $\varphi '_1, \dots , \varphi '_{{\mathcal N}'} $ by changing the first bead to the right of the division of the first Maya diagram to black.
We move the divisions of the first Maya diagrams of $\varphi '_1, \dots , \varphi '_{{\mathcal N}'} $ and that of $ \varphi '_1, \dots \varphi '_{{\mathcal N}'},\phi _0$ $m+1$-times to the right.
Then we recover $\varphi _1, \dots , \varphi _{{\mathcal N}}  $ from $\varphi '_1, \dots , \varphi '_{{\mathcal N}'} $ and we have $\varphi _1, \dots , \varphi _{\ell -1}, \varphi _{\ell +1} , \dots \varphi _{{\mathcal N}} $ from $\varphi '_1, \dots \varphi '_{{\mathcal N}'},\phi _0 $.
It follows from  $\text{W}[\varphi '_1, \dots , \varphi ' _{{\mathcal N}'}](x;g-m-1,h-m-1) \propto \text{W}[\varphi _1, \dots , \varphi _{{\mathcal N}}](x;g,h) (\sin x )^{(m+1) \{ g-(m+2)/2 \} } (\cos x)^{(m+1) \{ h-(m+2)/2 \} } $ that
\begin{align}
& V = U(x;g,h)-2\frac{d^2\log\text{W}[\varphi _1, \dots , \varphi  _{{\mathcal N}'}](x;g,h)}{dx^2} +4(m+1)(g+h-m-1) .
\end{align}
Combining with 
$\text{W}[\varphi '_1, \dots , \varphi ' _{{\mathcal N}'},\phi _0](x;g-m-1,h-m-1) \propto \\ \text{W}[\varphi _1, \dots , \varphi _{\ell -1}, \varphi _{\ell +1} , \dots \varphi _{{\mathcal N}} ](x;g,h)  (\sin x )^{(m+1) \{ g-(m+2)/2 \} } (\cos x)^{(m+1) \{ h-(m+2)/2 \} } $,\\ we have the proposition.
\end{proof}
We investigate the eigenvalues of the deformed PT Hamiltonian given by Eq.(\ref{eq:dPTHamil}) where the states $\varphi _1, \dots , \varphi _{{\mathcal N}}$ are described as Eq.(\ref{eq:tuple}), the deformed potential is non-singular on the open interval $(0,\pi /2)$ and $g$, $h$ are enough large.
It follows from Proposition \ref{prop:delstate} that the eigenvalue $ \mathcal{E}_{-m-1}=-4(m+1 ) (g+h-m-1)$ $(m \in \{ d^{\aitr }_{1}, \dots , d^{\aitr }_{M_{\aitr }} \}  )$ is permitted.
The eigenvalue $ \mathcal{E}_{n}=4n (g+h+n)$ $(n \in \Z _{\geq 0} \setminus \{ d^{\nul }_{1}, \dots , d^{\nul }_{M_{\nul }} \}  )$ is permitted and the eigenfunction is given by Eq.(\ref{eq:phiMn}).
The labeling of the eigenvalues corresponds to the position of the white beads of the first Maya diagram.

In the example of the deformed PT Hamiltonian given by the states $\tilde{\phi } ^{\ai }_{3} \tilde{\phi } ^{\ait }_{2} \tilde{\phi } ^{\aitr }_{1} \tilde{\phi } ^{\aitr }_{4}$ $ \tilde{\phi } ^{\aitr }_{5} \phi _{1} \phi _{3}$, the permitted  eigenvalues are $ \mathcal{E}_{-6} $, 
$ \mathcal{E}_{-5} $, $ \mathcal{E}_{-2} $, $ \mathcal{E}_{0} $, $ \mathcal{E}_{2} $, $ \mathcal{E}_{4} $, $ \mathcal{E}_{5} $, $ \mathcal{E}_{6} ,\dots $ and the first Maya diagram with division is 
$$
{{\dots \bullet \, \bullet \bullet \, \circ  \circ \, \bullet  \bullet \, \circ  \bullet | \circ   \bullet \circ  \bullet\circ \circ \circ \dots  }\atop{\dots \mbox{-}9 \mbox{-}8  \mbox{-}7 \mbox{-}6 \mbox{-}5 \mbox{-}4  \mbox{-}3 \mbox{-}2 \mbox{-}1\; \; 0\: 1\: 2\: 3 \: 4\: 5\: 6 \dots \; }}.
$$

 \section{Concluding remarks} \label{sec:rem}
In this article, we gave a correspondence between tuples of states and pairs of Maya diagrams with a division.
We have shown that a movement of the division corresponds to an identity of Wronskians of the states and that any Wronskian of four types of states is essentially equal  to a Wronskian of eigenstates and type $\ai $ seed solutions.
Here we propose a problem that the condition 
\begin{align}
 &\text{W}[\tilde{\phi } ^{\ai }_{d^{\ai }_1 } ,\ldots,\tilde{\phi } ^{\ai }_{d^{\ai }_{M_{\ai }}}, \tilde{\phi } ^{\nul }_{d^{\nul }_1},\ldots ,\tilde{\phi } ^{\nul }_{d^{\nul }_{M_{\nul }}} ] (x,g,h) \propto & \\
& \text{W}[\tilde{\phi } ^{\ai }_{\tilde{d}^{\ai }_1 } ,\ldots,\tilde{\phi } ^{\ai }_{\tilde{d}^{\ai }_{\tilde{M}_{\ai }}}, \tilde{\phi } ^{\nul }_{\tilde{d}^{\nul }_1},\ldots ,\tilde{\phi } ^{\nul }_{\tilde{d}^{\nul }_{\tilde{M}_{\nul }}} ] (x,g +m_1,h +m_2 ) (\sin x)^{m_3} (\cos x)^{m_4} , & \nonumber
\end{align}
$( 0< d^{\rm J }_1 < \dots < d^{\rm J }_{M_{\rm J }} ,  0< \tilde{d}^{\rm J }_1 < \dots <\tilde{d}^{\rm J }_{M_{\rm J }} , \: ({\rm J} =\ai , \nul ), \: m_1, m_2 \in \Z , \: m_3, m_4 \in \Rea )$ for any $g$ and $h$ leads to that the two tuples coincide and $m_1 =m_2=m_3=m_4=0$ or not?
If $g$ and $h$ are special, the above problem is negative because we have examples (e.g.  the systems equivalent to $\tilde{\phi } ^{\ai }_2 \tilde{\phi } ^{\aitr }_1 $ ($g=3(h-3)/(4h-9) $) and $\tilde{\phi } ^{\ai }_1 \tilde{\phi } ^{\aitr }_2 $ ($g=3h/(4h-9) $)) in \cite{HST}.

Our results in this article hold essentially true for the multi-indexed Laguerre polynomials.
We express the corresponding results for the multi-indexed Laguerre polynomials by using the notations in \cite{ST}.
We also associate tuples of eigenstates and three types of seed solutions Maya diagrams with couples of Maya diagrams with divisions.
Then a movement of the division corresponds to an identity of Wronskians of the states.
It is also shown that any Wronskian of four types of states is essentially equal to a Wronskian of eigenstates and type $\ai $ seed solutions. Namely we have
\begin{align}
&  \text{W}[\tilde{\phi } ^{\ai }_{d^{\ai }_1},\ldots,\tilde{\phi } ^{\ai }_{d^{\ai }_{M_{\ai }}}, \tilde{\phi } ^{\ait }_{d^{\ait }_1},\ldots,\tilde{\phi } ^{\ait }_{d^{\ait }_{M_{\ait }}},   \tilde{\phi } ^{\aitr }_{d^{\aitr }_1},\ldots,\tilde{\phi } ^{\aitr }_{d^{\aitr }_{M_{\aitr }}}, \phi _{d^{\nul }_1},\ldots, \phi _{d^{\nul }_{M_{\nul }}}] (x;g) \label{eq:thmainulL} \\
& \propto\text{W}[\tilde{\phi } ^{\ai }_{e^{\ait }_1 },\ldots,\tilde{\phi } ^{\ai }_{e^{\ait }_{\bar{M}_{\ait }}}, \tilde{\phi } ^{\ai }_{d^{\ai }_1 +d^{\ait }_{M_{\ait }} +1},\ldots,\tilde{\phi } ^{\ai }_{d^{\ai }_{M_{\ai }}+d^{\ait }_{M_{\ait }} +1}, \tilde{\phi } ^{\nul }_{e^{\aitr }_1 },\ldots,\tilde{\phi } ^{\nul }_{e^{\aitr }_{\bar{M}_{\aitr }}},  \tilde{\phi } ^{\nul }_{d^{\nul }_1 +d^{\aitr }_{M_{\aitr }} +1},\ldots \nonumber \\
& \qquad \qquad,\tilde{\phi } ^{\nul }_{d^{\nul }_{M_{\nul }}+d^{\aitr }_{M_{\aitr }} +1} ] (x,g-d^{\ait }_{M_{\ait }}  -d^{\aitr }_{M_{\aitr }} -2 ) x^{g_{\ai, \nul}} \exp ((d^{\aitr }_{M_{\aitr }}  -d^{\ait }_{M_{\ait }} )x^2/2 ), \nonumber
\end{align}
instead of Eq.(\ref{eq:thmainul}), where the notations in Theorem \ref{thm:red} are used.
Key relations are 
\begin{align}
& \text{W}[\tilde{\phi } ^{\ai }_{0}, \tilde{\phi } ^{J }_{n} ](x;g) \propto \tilde{\phi } ^{J }_{n -(\ai ,J) }  (x;g+1 ) \tilde{\phi } ^{\ai }_{0} (x;g) , \label{eq:W0JnL} \\
& \text{W}[\tilde{\phi } ^{\ait }_{0}, \tilde{\phi } ^{J }_{n} ](x;g) \propto \tilde{\phi } ^{J }_{n -(\ait ,J) }  (x;g-1) \tilde{\phi } ^{\ait }_{0} (x;g) , \nonumber \\
& \text{W}[\tilde{\phi } ^{\aitr }_{0}, \tilde{\phi } ^{J }_{n} ](x;g) \propto \tilde{\phi } ^{J }_{n -(\aitr ,J) }  (x;g-1 ) \tilde{\phi } ^{\aitr }_{0} (x;g) , \nonumber \\
& \text{W}[{\phi } _{0}, \tilde{\phi } ^{J }_{n} ](x;g) \propto \tilde{\phi } ^{J }_{n -(N ,J) }  (x;g+1) {\phi } _{0} (x;g) , \nonumber 
\end{align}
which corresponds to Eq.(\ref{eq:W0Jn}).

Odake and Sasaki \cite{osMW} extended the multi-indexed orthogonal polynomials to multi-indexed Wilson and Askey-Wilson polynomials, which appear in the discrete quantum mechanics, and Odake \cite{Odake} established that a multi-indexed Wilson (Askey-Wilson) polymonial which is expressed by type $\ai$ seed solutions and the type $\ait$ seed solutions is proportional to a multi-indexed Wilson (Askey-Wilson) polymonial which is expressed by only type $\ai$ seed solutions.
We believe that Odake's results are generalized and the results in this article are extended to the multi-indexed Wilson and Askey-Wilson polynomials.

\section*{Acknowledgements}
The author thanks Satoru Odake and Ryu Sasaki for valuable comments and fruitful discussions.
He also thanks C.-L. Ho for hospitality during his visit in Taipei in August 2013 where this work was started.
He is supported in part by the Grant-in-Aid for Young Scientists from the Japan Society for the Promotion of Science (JSPS), No.22740107.

\end{document}